\documentclass[12pt,onecolumn,draftclsnofoot]{IEEEtran}

\usepackage{cite}
\usepackage{psfrag}
\usepackage[pdf]{graphicx}
\usepackage[latin1]{inputenc}
\usepackage[T1]{fontenc}
\usepackage{amsmath,amsfonts,amsbsy,amssymb}
\usepackage{mathabx}
\usepackage{mathrsfs}
\usepackage[nolist]{acronym}
\usepackage{tabularx}
\usepackage{amssymb}
\usepackage{amsmath}
\usepackage{graphicx}
\usepackage[latin1]{inputenc}
\usepackage{graphicx}
\usepackage{cite}
\usepackage{multirow}
\usepackage{wasysym}
\usepackage{multirow}
\usepackage{float}
\usepackage{color}
\usepackage{enumitem}
\usepackage{amsthm}

\newtheorem{definition}{Definition}
\newtheorem{proposition}{Proposition}
\newtheorem{lemma}{Lemma}

\newtheorem{corollary}{Corollary}

\newcommand{\argmax}{\mathop{\mathrm{argmax}}}

\usepackage[colorlinks,bookmarksopen,bookmarksnumbered,linkcolor=black,citecolor=black,urlcolor=black]{hyperref}

\begin{document}
	
	\title{Enhanced Transmit Antenna Selection Scheme for Secure Throughput Maximization \\Without CSI at the Transmitter and its Applications on Smart Grids}
	\author{
		\IEEEauthorblockN{Hirley Alves, Mauricio Tomé, Pedro H. J. Nardelli,\\ Carlos H. M. de Lima and Matti Latva-aho} 
		\thanks{H. Alves, M. C. Tomé, P. H. J. Nardelli and M. Latva-aho are with the Centre for Wireless Communications (CWC), University of Oulu, Finland. Contact: [halves,mdecastr,nardelli,matla]@ee.oulu.fi. 
			C. H. M. de Lima is with São Paulo State University (UNESP), São João da Boa Vista, Brazil. Contact: carlos.lima@sjbv.unesp.br
			This work is partly funded by Finnish Academy (Aka) and CNPq/Brazil (n.490235/2012-3) through a joint project SUSTAIN, by SRC/Aka/BC-DC, and  by  European Commission through the P2P-SmarTest (n.646469).
		}
	}
	%
	\maketitle
	
	\begin{abstract}
		This paper addresses the establishment of secure communication links between smart-meters (Alice) and an aggregator (Bob) in the presence of an eavesdropper (Eve). The proposed scenario assumes: (i) MIMOME wiretap channel; (ii) transmit antenna selection at the Alice; (iii) no channel state information at the transmitter; (iv) fixed Wyner codes; and (v) guarantee of secure throughput by both quality of service and secrecy outage constraints. We propose a simple protocol to enhance security via transmit antenna selection, and then 
		assess its performance in closed-form by means of secrecy outage and successful transmission probabilities. 
		We assume these probabilities are our constraints and then maximize the secure throughput, establishing a security-reliability trade-off for the proposed scenario. 
		Our numerical results illustrate the effect of this trade-off on the secure throughput as well as on the number of antennas at Alice, Bob and Eve. 
		Interestingly, a small sacrifice in reliability allows secrecy enhancement in terms of secure bps/Hz. 
		We apply this idea in our smart grid application to exemplify that, although Eve may acquire some samples of the average power demand of a household, it is not enough to properly reconstruct such curve.
	\end{abstract}

	
	\section{Introduction}
	
	Wireless networks have become an indispensable part of our daily life through several applications that allows us to remotely monitor and control different processes within our homes, workplaces or even modern power grids. 
	In this context, each application has its own set of requirements and performance targets, which should be considered whenever designing communication systems. For instance, smart meters sending information about energy consumption have looser reliability and latency requirements than grid control applications in the high-voltage lines \cite{ART:Kuzlu-CN-2014, ART:Fang-CST2012}.
	
	One downside of wireless systems relates to information security and secrecy as they more susceptible to eavesdropping and denial of service attacks (e.g. jamming and spoofing) than wired systems due to its own nature \cite{ART:Wang-CN-2013, ART:Lee-CM-2012,ART:Shiu-WC-2011}. 
	To cope with such issue, current security systems are mainly based on cryptographic methods employed at the upper layers of communication protocols, while assuming limited computational power at the eavesdropper \cite{ART:Lee-CM-2012,ART:Shiu-WC-2011}.

	This assumption, nonetheless, is becoming an issue nowadays since the computational power of devices are steadily growing.
	Another weak point is that cryptographic solutions often overlook the physical properties of the wireless medium, the relative locations of the network elements and the actual transmitted information \cite{ART:Lee-CM-2012,ART:Shiu-WC-2011}.
	
	Information-theoretic security at the physical layer has reemerged to cope with such issues and complement cryptography by adding reliability and confidentiality at lower layers \cite{ART:Shiu-WC-2011}. 
	Physical layer (PHY) security can also open new ways to enhance robustness and reduce the complexity of conventional cryptography as far as it is built to be unbreakable and quantifiable (in confidential bps/Hz), regardless of the eavesdropper's computational power \cite{ART:Shiu-WC-2011}. 
	The notion of PHY-security was first introduced by Shannon in his seminal work in $1949$ \cite{ART:Shannon-1949}. But it was only later, in $1975$, that Wyner proposed in \cite{ART:Wyner-1975} the wiretap channel where the eavesdropper attempts to decode the information based on a degraded version of the legitimate link signal. 
	Later in \cite{ART:LeungYanCheong-1978}, authors showed the existence of a transmission rate that guarantees confidentiality based only on the statistics of the wireless channel.
	
	In $2008$, after a long period, those initial results are extended to account for the effects of fading channels \cite{ART:Liang-TIT-2008, ART:Bloch-TIT-2008}.
	Thereafter, different established techniques in wireless systems have been analyzed, for instance multiple antenna wiretap channel is characterized in \cite{ART:Khisti-TIT-2010b}, cooperative diversity is investigated in \cite{ART:Ding-JSAC-2012, ART:Alves-SPL-2015, ART:Brante-TC-2015}, while multi antenna diversity schemes are analyzed in \cite{PROC:Alves-Globecom2011, ART:Alves-SPL-2012, ART:Wang-TWC-2014}. 
	Notwithstanding all these fundamental results and advances, most works have quite restrictive assumptions on the eavesdropper, for instance, it is common to assume some (or even full) knowledge of the channel state information (CSI) of the eavesdropper \cite{ART:Shiu-WC-2011}, which turns out to be not feasible in practice since the legitimate transmitter may not be aware of the eavesdroppers.
	Alternatively, few works consider the case where no CSI is available at the transmitters \cite{ART:Tang-TIT-2009, ART:Zhou-CL-2011, ART:Brante-TC-2015}; however perfect secrecy cannot be achieved at all times and then secrecy outage characterization is performed in order to capture the probability of having a reliable and secure transmission.  
	
	Consequently, PHY-security is neglected as a suitable option, even when the application in consideration presents the characteristics that would make such an approach viable. Some applications of the smart energy grid are good examples where PHY-security appears as an attractive solution to enhance security and confidentiality \cite{nardelli2014models}.
	PHY-security enables an enhanced secure communication network (i) within smart-homes, (ii) between smart-meters and aggregators, and (iii) between aggregators and the local (cloud-)controller; and these three levels of communications are in fact the information backbone of the modern electricity distribution grid \cite{ART:Kuzlu-CN-2014,ART:Fang-CST2012, ART:Wang-CN-2013, ART:Lee-CM-2012}.
	Besides \cite{ART:Lee-CM-2012}, which summarizes the wireless network architecture in smart grid and proposes a key establishment protocol, few works consider PHY-security in this context. 
	
	In this work we attempt to fill this gap and focus on the secure communication between smart meters and aggregators in the presence of an eavesdropper (known as Eve). We assume that the smart meter poses as a legitimate transmitter (also known as Alice), while the aggregator acts as the legitimate receiver (known as Bob). 
	Both receivers (Bob and Eve) are able to estimate their own CSI, but Alice does not possess any CSI and resorts to adaptive encoder with constant transmit rate (which can be optimally chosen). 
	We build upon \cite{PROC:Alves-Globecom2011, ART:Alves-SPL-2012}, which introduces a scheme that allows only Bob to exploit diversity from Alice's transmission and thus limiting Eve's attack by design; therefore, we assume that all nodes have multiple antennas, but Alice employs transmit antenna selection (TAS) while Bob and Eve employ maximal ratio combining (MRC). 
	Then, we characterize the secrecy outage and secure throughput. Finally, we put our results in the context of smart grids, and thus resort to actual measured data and evaluate the impact of outages in the reconstruction of the average power demand by the aggregator. 

	Our results show that the proposed scheme achieves high reliability while restricting Eve capabilities by design and therefore enhancing security. Our main contributions are summarized next:
	\begin{itemize}
		\item we extend the results in \cite{ART:Alves-SPL-2012, ART:Zhou-CL-2011} by (i) assuming multi antenna wire-tap channel (all nodes have multiple antennas), (ii) characterizing in closed-form the secrecy outage probability for the case without CSI at the transmitter; (iii) provide a secure throughput analysis;
		\item we analyze the trade-off between security and reliability by introducing a parameter that reflects the quality of service of the legitimate link; 
		\item we propose a secure throughput maximization problem, and we evaluate the respective system performance with respect to the network configuration parameters;
		\item we investigate how the trade-off between secrecy and reliability affects performance in terms of secure throughput; 
		\item we apply our results to smart grids, resorting to actual data to support and exemplify our findings; we show that even if Eve acquires some information, it will not be enough to reconstruct the power demand curve of a household. 
	\end{itemize}
	
	The rest of this paper is organized as follows: Section~\ref{sc:system_model} introduces the system model and our main assumptions, Section~\ref{sc:out_tput} presents the secure outage probability analysis focusing on the optimization problem, and illustrates how the system performance changes with respect to the configuration parameters.
	Then, Section~\ref{sc:power} addresses the secure reconstruction of the average power demand curve as a function of the outage events, while Section~\ref{sc:discussions} discuss how our results might be used in actual deployments. Section~\ref{sc:conc} draws the final remarks and concludes this paper.
	
	\noindent{\textbf{Notation:}} Hereafter we denote scalar variables by italic symbols, while vectors and matrices are denoted by lower-case and upper-case boldface symbols, respectively. Given a complex vector $\mathbf{x}$, $||\mathbf{x}||$ denotes the Euclidean norm, then $(\mathbf{x})^T$ and $(\mathbf{x})^\dag$ denote transpose and conjugate transpose operations, respectively. 
	The $m\times m$ identity matrix is represented as $\mathbf{I}_m$. Probability density function (PDF) and cumulative distribution function (CDF) of a given random variable $X$ are denoted as $f_X(x)$ and $F_X(x)$, respectively, while its expectation is denoted as $\mathbb{E}\left[\cdot\right]$.
	Gamma function is defined as $\Gamma(z)$\cite[Ch 6, \S6.1.1]{BOOK:ABRAMOWITZ-DOVER03}, and the regularized lower incomplete gamma function is denoted as $\operatorname{P}\left(s , z\right) = \tfrac{\gamma(s,z)}{\Gamma(z)}$ \cite[Ch 6, \S6.5.1]{BOOK:ABRAMOWITZ-DOVER03}.
	
	\section{System model} \label{sc:system_model}
	
	We assume multiple antenna wiretap channel where a legitimate pair attempts to communicate securely in the presence of an eavesdropper. 
	The transmitter is known as Alice and represents a smart meter and possesses $N_A$ antennas, while the $N_B$ antenna receiver is known as Bob and plays as an aggregator (responsible for acquiring information from smart meters and performing control and management actions \cite{ART:Strbac-EP2008}). 
	\begin{figure}[!b]
		\centering
		\includegraphics[width=.6\columnwidth]{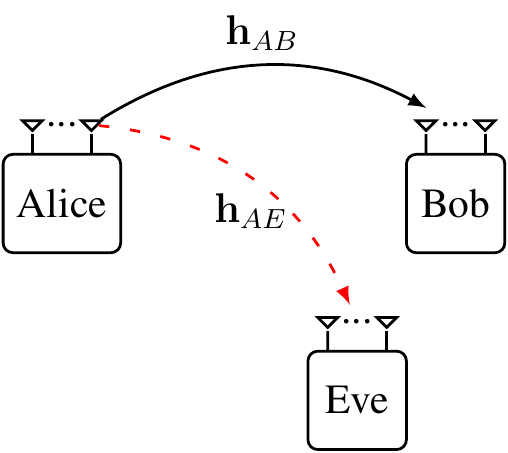}
		\caption{Network deployment illustration: Alice employs TAS, while Bob and Eve resort to MRC, but only Bob is able to exploit diversity from Alice's antennas. An error-free open channel is assumed between Bob and Alice, so that Bob can enable the Alice's transmission and inform the best antenna index.}
		\label{fig:system_model}
	\end{figure}
	The untrusted node, commonly named as Eve, is assumed to have $N_E$ antennas. Eve may eavesdrop and attempt to acquire data from Alice's transmissions. 
	Herein, we assume that Alice sends its average power demand to the aggregator, which by its turn reconstructs this signal in order to perform control and power demand management of its grid. 
	
	This scenario is depicted in Fig.~\ref{fig:system_model}, where the solid black line represents the communication between Alice and Bob, while eavesdropper link is depicted in as a dashed red arrow. 
	Moreover, both receivers are able to estimate their own CSI, but no CSI is fed back to Alice.
	However, there is a open and error-free feedback channel between Bob and Alice which is used to convey the index of Alice's antenna with the best signal-to-noise-ratio (SNR) and enable on-off transmission.

	As in \cite{ART:Alves-SPL-2012, ART:Brante-TC-2015}, such channel is open and error-free, and even if Eve acquires this feedback and knows the antenna index an optimum TAS scheme with respect to Bob is a random TAS scheme concerning Eve. Therefore, Eve is not able to exploit diversity from Alice's multiple antennas since legitimate and eavesdropper channels are uncorrelated.
	Another advantage of this approach is that the feedback channel can have limited capacity, and the number of bits necessary for this channel is $n_{bits}=\lceil \log_2 N_A \rceil$.

	\subsection{Transmission protocol and encoding scheme}
	
	The aggregator schedules and requests each smart meter to send its average power demand. Such a request is performed though the feedback channel, which not only carries the signaling to start the transmission but also the antenna index. 
	Since no CSI is fed back to Alice, we resort to fixed Wyner codes, with constant transmission rate, which can be optimally chosen given the network configuration parameters as we shall see in the next section.
	
	Let us first define the capacity of the legitimate and eavesdropper links as $C_b$ and $C_e$, respectively. Then, Bob chooses two rates: a transmission rate $R_b$ and a confidential rate $R_s$, and we define the cost of securing a transmission as $R_e = R_b-R_s$ \cite{ART:Tang-TIT-2009,ART:Zhou-CL-2011}. 
	Then, two conditions arise in order to guarantee secrecy and reliability: (i)
	whenever $C_b > R_b$ the message is correctly decoded at Bob; and (ii) an information leakage occurs whenever $C_e > R_e$ \cite{ART:Tang-TIT-2009,ART:Zhou-CL-2011}. 
	These conditions guarantee that there is a Wyner code that ensures a reliable (small error probabilities) and secure communication link. Further details of fixed Wyner codes and code construction can be found in \cite{ART:Tang-TIT-2009}. 
	Furthermore, in this context we adopt a probabilistic measure of security, namely \textit{secrecy outage probability}, and then we are able to characterize the secure throughput maximization problem analyzed in Section~\ref{sc:out_tput}. 
	
	\subsection{Legitimate and Eavesdropper Channel models}
	
	We assume that all channels coefficients are independent and the squared-envelope is exponentially distributed, thus we consider Rayleigh fading. 
	In the legitimate channel, a single transmit antenna is selected at Alice to maximize the SNR at Bob, which applies MRC at the received signal. The best antenna index is defined as $i^*$: 
	\begin{align}
	i^* = \argmax_{1 \leq~i~\leq N_A} ||\mathbf{h}_{iB}||,
	\end{align}
	where $\mathbf{h}_{iB} = \left[{h}_{i1}, {h}_{i2},\cdots,{h}_{i{N_B}} \right]^T$ denotes the $N_B \times 1$ channel vector between the $i$th transmit antenna at Alice and the $N_B$ antennas at Bob with independent and identically distributed (i.i.d.) Rayleigh fading. 
	
	Then, Alice encodes the message with the codeword $\mathbf{x}=\left[x(1,\cdots,x(i),\cdots,x(n))\right]$, using the aforementioned Wyner codes \cite{ART:Tang-TIT-2009}.
	We also assume that the codeword transmitted is subject to an average power constraint $\tfrac{1}{n}\sum_{i=1}^{n} \mathbb{E}\left[ |x(i)|^2 \right] \leq P_A$, where $P_A$ denotes Alice's transmit power. Then, Bob combines the signal vectors using MRC, which yields the following received signal at time $i$:
	\begin{align}
	\label{eq:yB}
	y_B(i) = \mathbf{h}_{AB}^\dag \mathbf{h}_{AB} x(i) + \mathbf{h}_{AB}^\dag \mathbf{n}_{AB},
	\end{align}
	where $\mathbf{h}_{AB} = \mathbf{h}_{{i^*}B}$ represents the legitimate channel vector, $\mathbf{n}_{AB}$ is the $N_B \times 1$ additive white Gaussian noise vector at Bob, assuming $\mathbb{E}\left[ \mathbf{n}_{AB}\mathbf{n}_{AB}^\dag \right] = \mathbf{I}_{N_B} \sigma_{AB}^2$, with $\sigma_{AB}^2$ being the noise variance at each antenna. 
	Thus, from \eqref{eq:yB} the instantaneous SNR of the legitimate link is 
	\begin{align}
	\gamma_B = \frac{||\mathbf{h}_{AB}||^2 P_A}{\sigma_{AB}^2},
	\end{align}
	and its PDF and CDF are defined, respectively, as \cite{ART:Chen-TWC-2009}
	\begin{align}
	\label{eq:pdf_gammaB}
	f_{\gamma_B}(\gamma) &= \frac{N_A \ \gamma^{N_B-1}}{\Gamma(N_B) \ \overline{\gamma}_B^{N_B}} \exp\left(-\frac{\gamma}{\overline{\gamma}_B}\right) \, \operatorname{P}\left(N_B \,,\, \frac{\gamma}{\overline{\gamma}_B}\right)^{N_A-1}, \\
	\label{eq:cdf_gammaB}
	F_{\gamma_B}(\gamma) &= \operatorname{P}\left( N_B \,,\, \frac{\gamma}{\overline{\gamma}_B} \right)^{N_A},
	\end{align}
	where $\overline{\gamma}_B$ denotes the average SNR and we recall that $\operatorname{P}\left(\cdot , \cdot \right)$ denotes the regularized lower incomplete gamma function \cite[Ch 6, \S6.5.1]{BOOK:ABRAMOWITZ-DOVER03}
	Notice from \eqref{eq:pdf_gammaB} and \eqref{eq:cdf_gammaB} that the legitimate channel exploits diversity from Alice and Bob's multiple antennas. 
	
	On the other hand, Eve perceives a random TAS scheme, thus can only exploit diversity from its own antennas. 
	Therefore, Eve combines the eavesdropped signal vectors using MRC, which yields the following received signal at time $i$
	\begin{align}
	\label{eq:yE}
	y_E(i) = \mathbf{h}_{AE}^\dag \mathbf{h}_{AE} x(i) + \mathbf{h}_{AE}^\dag \mathbf{n}_{AE},
	\end{align}
	where $\mathbf{h}_{AE} = \mathbf{h}_{{i^*}B}$ represents the eavesdropper channel vector, $\mathbf{n}_{AE}$ is the $N_E \times 1$ additive white Gaussian noise vector at Eve, assuming $\mathbb{E}\left[ \mathbf{n}_{AE}\mathbf{n}_{AE}^\dag \right] = \mathbf{I}_{N_E} \sigma_{AE}^2$, with $\sigma_{AE}^2$ being the noise variance at each antenna. 
	Similarly to the legitimate link, all channels undergo Rayleigh fading. In this context, we write the instantaneous SNR at Eve as $\gamma_E = \tfrac{||\mathbf{h}_{AE}||^2 P_A}{\sigma_{AE}^2}$, which follows Gamma distribution, and its PDF and CDF are given receptively as \cite{ART:Alves-SPL-2012}
	\begin{align}
	\label{eq:pdf_gammaE}
	f_{\gamma_E}(\gamma) &= \frac{\gamma^{N_E-1}}{\Gamma(N_E)\, \overline{\gamma}_E^{N_E}} \exp\left( - \frac{\gamma}{\overline{\gamma}_E}  \right), \\
	\label{eq:cdf_gammaE}
	F_{\gamma_E}(\gamma) &= \operatorname{P}\left(N_E \,,\, \frac{\gamma}{\overline{\gamma}_E}\right), 
	\end{align}
	where $\overline{\gamma}_E$ denotes the average SNR at Eve.

	\section{Secrecy Outage and Secure Throughput} \label{sc:out_tput}
	\begin{figure*}[!t]
		\centering
		\includegraphics[width=\columnwidth]{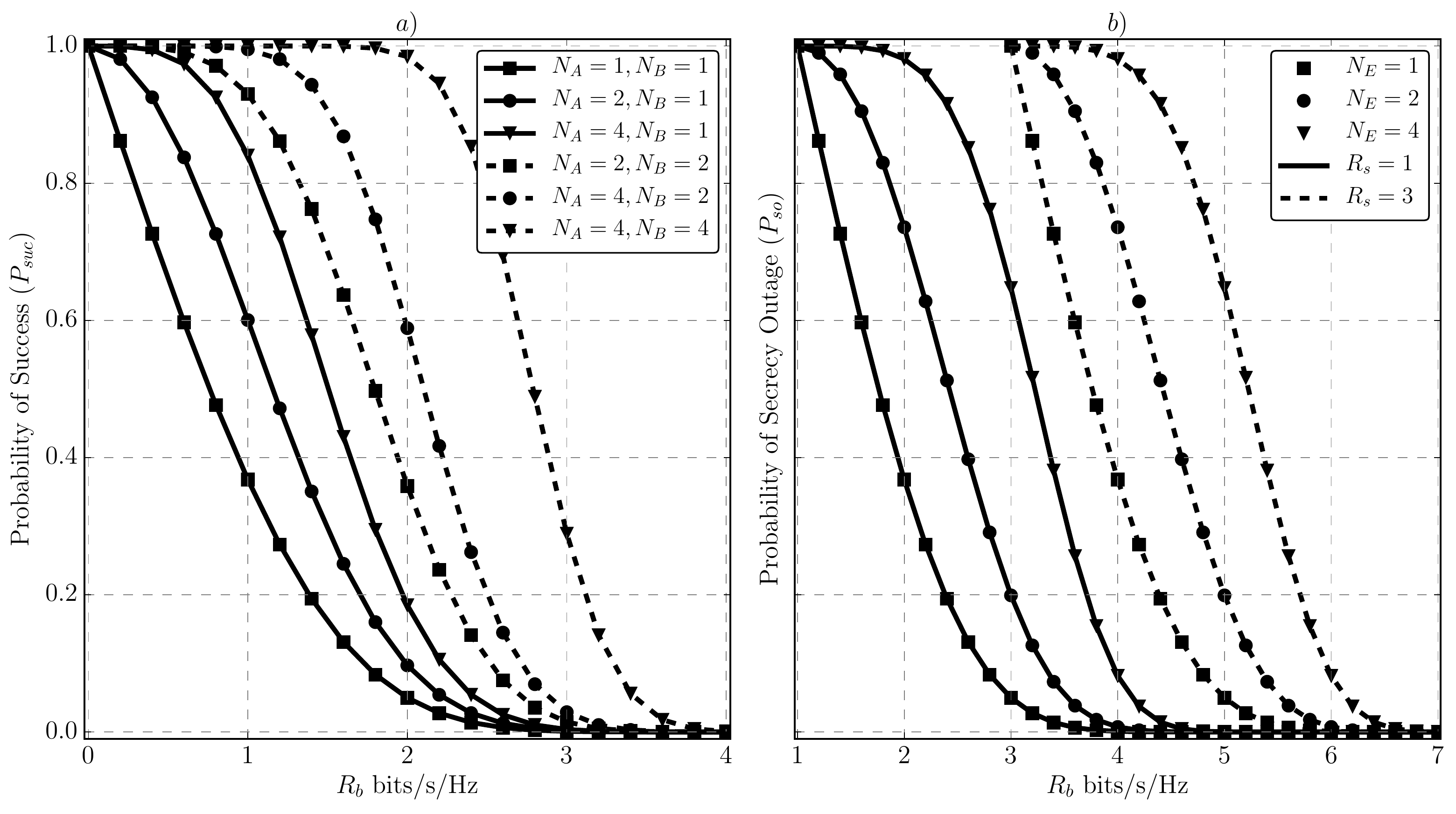}
		\caption{Example of the success and secrecy outage probabilities vs. the transmission rate $R_b$: $a)$ distinct antenna arrangements at the legitimate link with $\overline{\gamma}_B = 10$~dB; $b)$ secrecy outage for different number of antennas and for two secure rates $R_s \in \{1,3\}~\mathrm{bits/s/Hz}$ with $\overline{\gamma}_E = 0$~dB}
		\label{fig:psuc_pso}
		\vspace{-3ex}
	\end{figure*}
	As discussed above there are two conditions so as to guarantee secrecy and reliability \cite{ART:Tang-TIT-2009,ART:Zhou-CL-2011}. With respect to the former, the channel capacity has to be greater than the transmission rate, thus $C_b > R_b$ which ensures that the message is decoded. Therefore, we define the probability of successful transmission for the proposed scheme in the following lemma.
	\begin{lemma}
		\label{lm:psuc}
		The probability of successful transmission for the system model of Section~\ref{sc:system_model} assuming that an on-off transmission scheme, which occurs whenever $\gamma_B$ exceeds an SNR threshold $\mu$, is $p_{suc} = \Pr\left[ C_b> R_b \right] = \Pr\left[ \gamma_B > \mu \right] = 1-F_{\gamma_B}(\mu)$, where $F_{\gamma_B}(\cdot)$ is given in \eqref{eq:cdf_gammaB} and $\mu\geq 2^{R_b}-1$, which reflects the minimum value that guarantees reliability at the legitimate link. 
	\end{lemma}
	
	On the other hand, regarding security, an information leakage occurs whenever $C_e > R_e$, where $R_e = R_b - R_s$, and thus we have secrecy outage which can be defined as follows.
	\begin{lemma}
		\label{lm:pso}
		Given the system model of Section~\ref{sc:system_model} and fixed Wyner codes, the probability of secrecy outage is $p_{so} = \Pr\left[ C_e > R_b-Rs \right] = \Pr\left[ \gamma_E > 2^{R_b-R_s} -1 \right] = 1 - F_{\gamma_E}(2^{R_b-R_s} -1)$, where $F_{\gamma_E}(\cdot)$ is given in \eqref{eq:cdf_gammaE}.
	\end{lemma}
	
	Let us introduce an example of Lemmas \ref{lm:psuc} and \ref{lm:pso}. Fig.~\ref{fig:psuc_pso} illustrates the performance of the success probability  ($p_{suc}$) and secrecy outage probability  ($p_{so}$) as a function of the  transmission rate $R_b$. 
	As expected, the performance improves by increasing the number of antennas either at the legitimate link or at the Eve. However, Eve can only change its own diversity, and thus outage probability (we recall that higher secrecy outage, means that more information is acquired by the eavesdropper), by adding more antennas, since it cannot exploit diversity from Alice's antennas. In its turn, the legitimate channel performance enhances even more if the aggregator dedicates more antennas to reception. Additionally, notice that we assume a multiple antenna scenario, encompassing the single antenna case introduced in \cite{ART:Zhou-CL-2011}. 
	
	After presenting Lemmas \ref{lm:psuc} and \ref{lm:pso}, we are able to define the secure throughput and the maximization problem.
	\begin{definition}[Secure throughput] \label{df:tput}
		The	secure throughput $T_s$ of the legitimate link (between smart meter and aggregator) is defined as 
		\begin{align}\label{eq:tput}
		T_s \stackrel{\small \triangle}{=} R_s \, p_{suc} = R_s \, \left(1 - \operatorname{P}\left( N_B \,,\, \frac{\mu}{\overline{\gamma}_B} \right)^{N_A} \right).
		\end{align}
	\end{definition}
	
	With respect to Alice, our goal is to determine the best transmission rate that ensures both reliability and secrecy.
	This thus maximizes the secure throughput to Bob, while respecting secrecy outage ($p_{so}(R_b, R_s) \leq \epsilon$) and QoS constrains ($p_{suc}(\mu) \geq \sigma$). 
	Then, we can define the following maximization problem as
	\vspace{-1ex}
	\begin{align}\label{eq:opt_problem1}
	\begin{aligned}
	& \argmax_{R_s, R_b, \mu}
	& & T_s \\
	& \mathrm{subject~to}
	& & p_{so}(R_b, R_s) \leq \epsilon \\
	& & & p_{suc}(\mu) \geq \sigma \\
	& & & \mu \geq 2^{R_b}-1 \\
	& & & R_s > 0,
	\end{aligned}
	\vspace{-2ex}
	\end{align}
	where $0\leq \sigma \leq 1$ is the minimum acceptable success probability, reflecting the QoS constraint on the legitimate channel, and $0 \leq \epsilon \leq 1$ is the maximum acceptable information leakage. 

	Note that Alice is aware that eavesdropping may occur, and thus protects its transmission by optimally selecting a proper rate while minimizing the secrecy outage. We will further discuss the impact of these assumptions in Section~\ref{sc:power}. 
	From \eqref{eq:opt_problem1} and Lemma~\ref{lm:pso}, we can see that $p_{so}(R_b, R_s)$ is independent of $\mu$. Thereby, we first maximize $p_{suc}(\mu)$, which is monotonically decreasing with respect to $\mu$, by minimizing $\mu$. Hence, its optimal value is $\mu=2^{R_b}-1$. 
	\vspace{-1ex}
	\begin{proposition}
		\label{prop:rate_b}
		Assuming optimal $\mu=2^{R_b}-1$, the transmission rate $R_b$ that ensures $p_{suc}(\mu) \geq \sigma$ is
		\vspace{-1ex} 
		\begin{align}\label{eq:R_B}
		R_b \leq \log_2 \left(1+ \overline{\gamma}_B \, \alpha \, \log\left( \left( 1- (1-\sigma)^{\tfrac{1}{N_A \, N_B}}\right)^{-1} \right) \right),
		\end{align}
		\vspace{-1ex}
		where $\alpha = \Gamma(N_B+1)^{\tfrac{1}{N_B}}$.
	\end{proposition}
	\begin{proof}
		From Lemma~\ref{lm:psuc} we have $p_{suc}(\mu)$ and then under the constraint $\sigma$, we attain
		\begin{align}
		1-\operatorname{P}\left( N_B \,,\, \tfrac{\mu}{\overline{\gamma}_B} \right)^{N_A} &\geq \sigma \\
		\operatorname{P}\left( N_B \,,\, \tfrac{\mu}{\overline{\gamma}_B} \right) & \stackrel{(a)}{\leq} (1-\sigma)^{\tfrac{1}{N_A}} \\
		\left(1- \exp\left( -\dfrac{\mu}{\overline{\gamma}_B \, \alpha} \right) \right)^{N_B} &\stackrel{(b)}{\leq} (1-\sigma)^{\tfrac{1}{N_A}} \\
		\mu &\stackrel{(c)}{\leq} \overline{\gamma}_B \, \alpha \, \log\left( \xi^{-1} \right),
		\end{align}
		\begin{enumerate}[label=$(\alph*)$]
			\item since $0 \leq \sigma \leq 1$ we isolate the regularized gamma function, which is invertible only for the equality, thus $\mu = \overline{\gamma}_B \operatorname{P}^{-1}\left( N_B \,,\, (1-\sigma)^{\tfrac{1}{N_A}} \right)$, where $\operatorname{P}^{-1} (a,x)$ is the inverse of the generalized regularized incomplete gamma function defined in \cite{URL:Wolfram-InvGamma, URL:scipy-invgamma}; otherwise,
			\item since $N_B>0$ $\tfrac{\mu}{\overline{\gamma}_B}>0$, we rewrite $(a)$ by resorting to the following inequality $(1-\exp(-\alpha_a x))^a \leq \operatorname{P}(a, x)$, where $\alpha_a = \Gamma(1+a)^{1/a}$ (equality holds for $a=1$) \cite[Ch8, \S8.10.11]{BOOK:Olver-2010},
			\item last, since all variables are positive we isolate the variable $\mu$, where $\xi = \left( 1- (1-\sigma)^{\tfrac{1}{N_A \, N_B}}\right)$
		\end{enumerate}
		Finally, we know that $\mu=2^{R_b}-1$, thus we readily attain \eqref{eq:R_B}. 
	\end{proof} 
	\begin{corollary}
		Assuming $N_A \in \mathbb{Z}^*$ and $N_B = 1$, which is the case when only TAS is employed at the legitimate channel, then $R_b \leq \log_2 \left(1+ \overline{\gamma}_B \, \log\left( \left( 1- (1-\sigma)^{\frac{1}{N_A}}\right)^{-1} \right) \right)$. While for $N_A=N_B=1$ (single antenna case), \eqref{eq:R_B} reduces to $R_b \leq \log_2 \left(1+ \overline{\gamma}_B \, \log\left( \sigma^{-1} \right) \right)$ as in \cite{ART:Zhou-CL-2011}. 
	\end{corollary}

	\begin{figure*}[!t]
		\centering
		\includegraphics[width=\columnwidth]{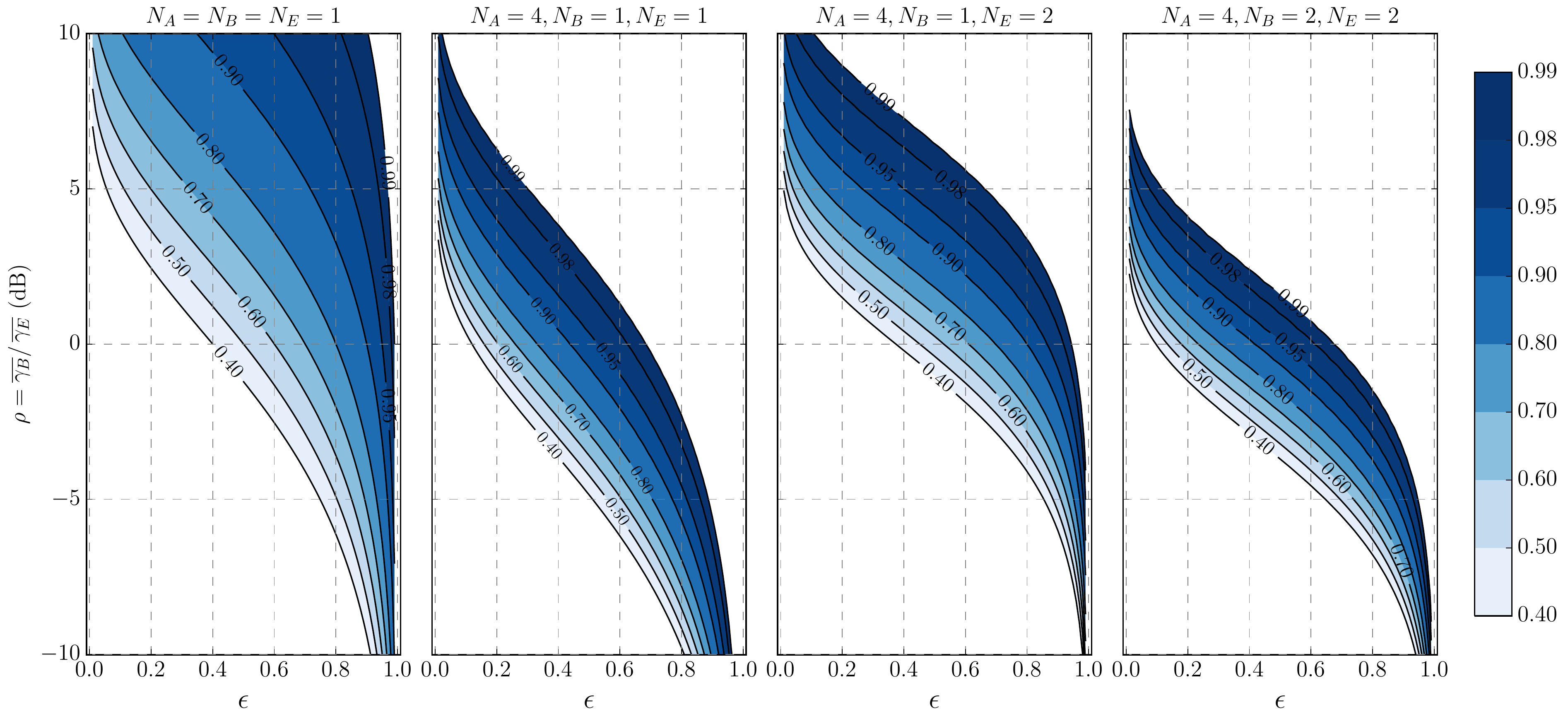}
		\caption{Illustrative example of the trade-off between reliability and security for distinct sets of antennas, from the single (left) to multiple (right) antennas. Contour plots indicate the value of $.4 \leq \sigma \leq 0.99$ as a function of the relative gain between the average SNR of the legitimate and eavesdropper channels, namely $\rho = \tfrac{\overline{\gamma}_B}{\overline{\gamma}_E}$, and $\epsilon$.}
		\label{fig:sigma_epsilon}
		\vspace{-3ex}
	\end{figure*}

	Next, we tackle the restriction on the information leakage $p_{so}(R_b, R_s) \leq \epsilon$. 
	\begin{proposition}
		\label{prop:r_s}
		For any $R_b>R_s$ the secrecy outage is monotonically decreasing with $R_b$, while monotonically increasing with respect to $R_s$. Thus, satisfying $p_{so}(R_b, R_s) \leq \epsilon$, the throughput maximizing $R_s$ is
		\begin{align}
		\label{eq:R_S}
		R_s &= R_b - \log_2\left(1+\overline{\gamma}_E \operatorname{P}^{-1}\left( N_E \,,\, 1-\epsilon \right) \right),
		\end{align}
		where $\operatorname{P}^{-1} (a,x)$ is the inverse of the generalized regularized incomplete gamma function \cite{URL:Wolfram-InvGamma}\footnote{It is noteworthy that $\operatorname{P}^{-1} (a,x)$  is an analytic function of $a$ and $x$ and can be easily evaluated through standard mathematical frameworks such as Mathematica \cite{URL:Wolfram-InvGamma} as well as SciPy \cite{URL:scipy-invgamma}.}.
	\end{proposition}
	\begin{proof}
		Since for any $R_b>R_s$ the secrecy outage is a decreasing function of $R_b$, the maximizing throughput $R_s$ occurs when $p_{so}(R_b, R_s) = \epsilon$. In the equality $\operatorname{P} (a,z)=x$ is invertible \cite{URL:Wolfram-InvGamma}, which allow us to isolate $z = 2^{R_b-R_s}-1$ and then attain $R_s$ as in \eqref{eq:R_S}.
	\end{proof}
	
	We are about to state the simplified version of our maximization problem given Propositions \ref{prop:rate_b} and \ref{prop:r_s} discussed above. But first, let us introduce an important result that allows us to assess the trade-off between reliability and security.  
	\begin{proposition}
		\label{prop:sec_reliability}
		Given a positive secrecy rate $R_s >0$, we establish the trade-off between reliability and security as 
		\begin{align} \label{eq:tradeoff_s_e}
		\sigma < 1 - \left(1 - \exp\left(- \dfrac{ \operatorname{P}^{-1}\left(N_E, 1-\epsilon\right)  }{\rho \, \alpha } \right) \right)^{N_A N_B}, 
		\end{align} 
		where $\rho \stackrel{\small \triangle}{=} {\overline{\gamma}_B}/  {\overline{\gamma}_E}$ defines the relative gain between the average SNR of the legitimate (${\overline{\gamma}_B}$) and eavesdropper (${\overline{\gamma}_E}$) channels. 
	\end{proposition}
	\begin{proof}
		In order to achieve a positive secrecy rate $R_s>0$, we have to guarantee that $R_b > \log_2\left(1+\overline{\gamma}_E \operatorname{P}^{-1}\left( N_E \,,\, 1-\epsilon \right) \right)$. From \eqref{eq:R_B} we attain $R_b$, and then we isolate $\sigma$ as follows
		\begin{align}
		\hspace{-1ex}-\overline{\gamma}_B \alpha \log\left(1- (1-\sigma)^{\tfrac{1}{N_A \, N_B}}\right) &> \overline{\gamma}_E \operatorname{P}^{-1}\left( N_E \,,\, 1-\epsilon \right) \\
		\hspace{-1ex} \log\left(  1- (1-\sigma)^{\tfrac{1}{N_A \, N_B}} \right)  &\stackrel{(a)}{<} \dfrac{\operatorname{P}^{-1}\left( N_E \,,\, 1-\epsilon \right)}{\rho \, \alpha },
		%
		%
		\end{align}
		\begin{enumerate}[label=$(\alph*)$]
			\item since $0 \leq \epsilon \leq 1$, $N_A>0,~N_B>0$, we isolate the function of $\sigma$ in the right-side and then define $\rho$; then		
		\end{enumerate}
		since all variables are positive and grater than zero, we perform some algebraic manipulations and isolate $\sigma$ as in \eqref{eq:tradeoff_s_e}.
	\end{proof}
	\begin{corollary}
		Assuming $N_A \in \mathbb{Z}^*$ and $N_B = 1$, then $\sigma < 1 - \left(1 - \exp\left(- \operatorname{P}^{-1}\left(N_E, 1-\epsilon\right)  /\rho \right) \right)^{N_A} $, whilst for the single antenna case, \eqref{eq:R_B} reduces to $\sigma < \epsilon^{1/\rho}$, which was also attained in \cite{ART:Zhou-CL-2011}. 
	\end{corollary}

	Fig.~\ref{fig:sigma_epsilon} illustrates the trade-off between reliability and security stated in Proposition~\ref{prop:sec_reliability}. We evaluate $\sigma$, which can be seen as a QoS/reliability indicator, as a function of $\epsilon$, which denotes how much secrecy outage the system tolerates, as well as $\rho$, which captures how good is the main channel with respect to the eavesdropper's channel. 
	Fig.~\ref{fig:sigma_epsilon} has four settings: from single to multi antenna configuration.
	
	For instance, if Alice employs only TAS ($N_A=4$), Bob and Eve are single antenna, there is a great gain in reliability with respect to the single antenna case, in fact, reliability grows from 60\% to about 97.5\% for $\epsilon=.2$ and $\rho=5$~dB. However, as $N_E$ increases the feasibility region diminishes. For instance, if Eve has one more antenna, thus $N_E=2$, $\sigma$ drops from 97.5\% to about 86\%.
	This effect can be counteracted by adding more antennas to the legitimate link, thus enhancing reliability through diversity. This case is exemplified on the rightmost plot of Fig.~\ref{fig:sigma_epsilon}, where Bob now has $N_B=2$ antennas, which renders more than 99\% of reliability for $\epsilon>0.1$ and $\rho>5$~dB.

	In this discussion we set $\epsilon = 0.2$, which is somewhat a high value for secrecy constraints. As we shall discuss in the next section, such high secrecy outage constraint may be feasible (acceptable) depending on the application. Of course, the less information lost the better, especially if the information is critical. We recall that herein we are only evaluating security at PHY layer as a way to complement some cryptographic method implemented in the higher layers of the protocol stack. Nonetheless, our results also show ways to increase security at PHY layer, thus smaller values of $\epsilon$, by increasing the number of antennas as well as guaranteeing high SNR at the main link (larger $\rho$). 
	
	Finally, one more way to increase performance of the legitimate link is to maximize the secure throughput, which is hereby our goal and we are now ready to state the simplified version of our maximization problem given Propositions \ref{prop:rate_b} to \ref{prop:sec_reliability}. 
	Therefore, the secure throughput maximization problem is rewritten as,
	\begin{align} \label{eq:opt_problem1}
	\begin{aligned}
	& \argmax_{R_b} 
	& & T_s = \left(R_b - R_e \right)\,\left( 1- \operatorname{P}\left( N_B \,,\, \frac{2^{R_b}-1}{\overline{\gamma}_B} \right)^{N_A}  \right) \\
	& \mathrm{subject~to}
	& & R_e < R_b 
	\end{aligned}
	\end{align}
	where $R_b$ is given in \eqref{eq:R_B} from Proposition~\ref{prop:rate_b}, and 
	$R_e = \log_2\left(1+\overline{\gamma}_E \operatorname{P}^{-1}\left( N_E \,,\, 1-\epsilon \right) \right)$ comes from Proposition \ref{prop:r_s}.
	\begin{proposition} 
		\label{prob:opt_rb}
		The optimal secure throughput of our proposed scheme is given as
		\begin{align}
		\label{eq:tput_opt}
		T_s^* = \left(R_b^* - R_e \right)\,\left( 1- \operatorname{P}\left( N_B \,,\, \frac{2^{R_b^*}-1}{\overline{\gamma}_B} \right)^{N_A}  \right),
		\end{align}
		where the optimal transmission rate $R_b^*$ is the solution for the following transcendental equation
		\begin{align}
		\label{eq:r_b_opt}
		1-\operatorname{P}\left( N_B \,,\, y \right)^{N_A}\!&=\! \beta y^{N_B-1} e^{-y} \operatorname{P}\left( N_B \,,\, y \right)^{N_A-1},
		\end{align}  
		where $y = \tfrac{2^{R_b}-1}{\overline{\gamma}_B}$ given the domain $R_e < R_b$ and respecting the condition \eqref{eq:tradeoff_s_e}, where 
		$\beta = \frac{\log(2) N_A (R_b-R_e) 2^{R_b}}{\Gamma(N_B) \overline{\gamma}_B}$.
	\end{proposition}
	\begin{proof}
		The function $T_s$ is continuous and concave in the domain $R_e < R_b$ (with $N_A, N_B \in \mathbb{Z}^*$ and $\overline{\gamma}_B>0$), where $R_b$ is given in \eqref{eq:R_B}, since its second derivative with respect to $R_b$ is negative, thus ${\partial^2 T_s}/{\partial R_b^2} < 0$, therefore $R_b^*$ is attained by solving the first derivative of $T_s$ with respect to $R_b$ and equating to zero, ${\partial T_s}/{\partial R_b} = 0$, which after some algebraic manipulations yields \eqref{eq:r_b_opt}.
		
		Unfortunately, \eqref{eq:r_b_opt} does not have a closed-form expression, though it is noteworthy that \eqref{eq:r_b_opt} can be easily evaluated numerically using mathematical frameworks such as Mathematica and SciPy \cite{URL:sympy-nsolve}. For the single antenna case an closed-from expression for $R_b^{*}$ can be attained as in \cite{ART:Zhou-CL-2011}. 
		%
	\end{proof}
	\begin{figure}[!b]
		\vspace{-2ex}
		\centering
		\includegraphics[width=\columnwidth]{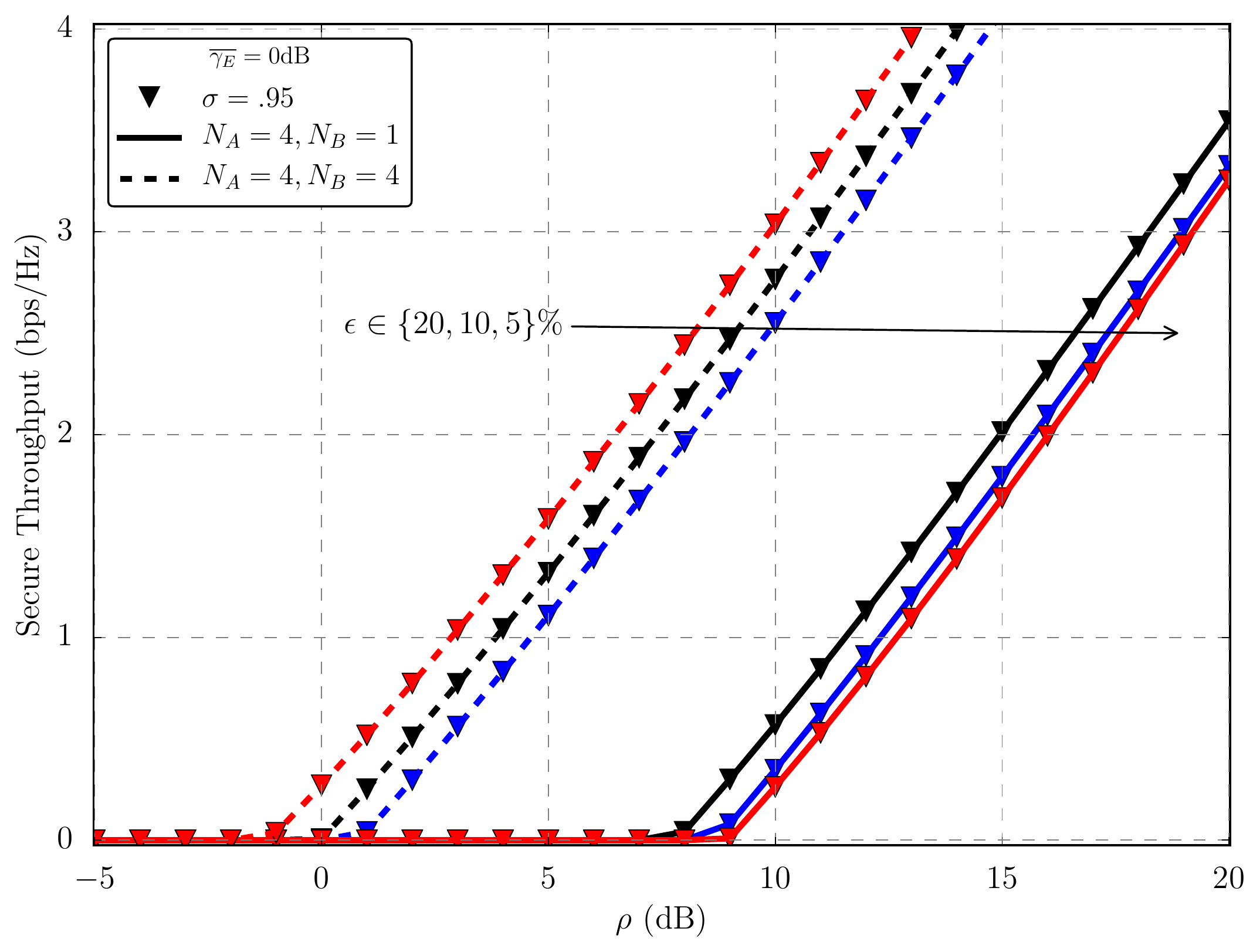}
		\caption{Secure throughput as a function of the SNR of the legitimate link $\overline{\gamma}_B$, assuming $\overline{\gamma}_E=0$~dB, $\sigma=95\%$ and distinct antenna configurations and secrecy outage thresholds.}
		\label{fig:tput_gamma}
	\end{figure}
	
	\begin{figure*}[!t]
		\centering
		\includegraphics[width=\columnwidth]{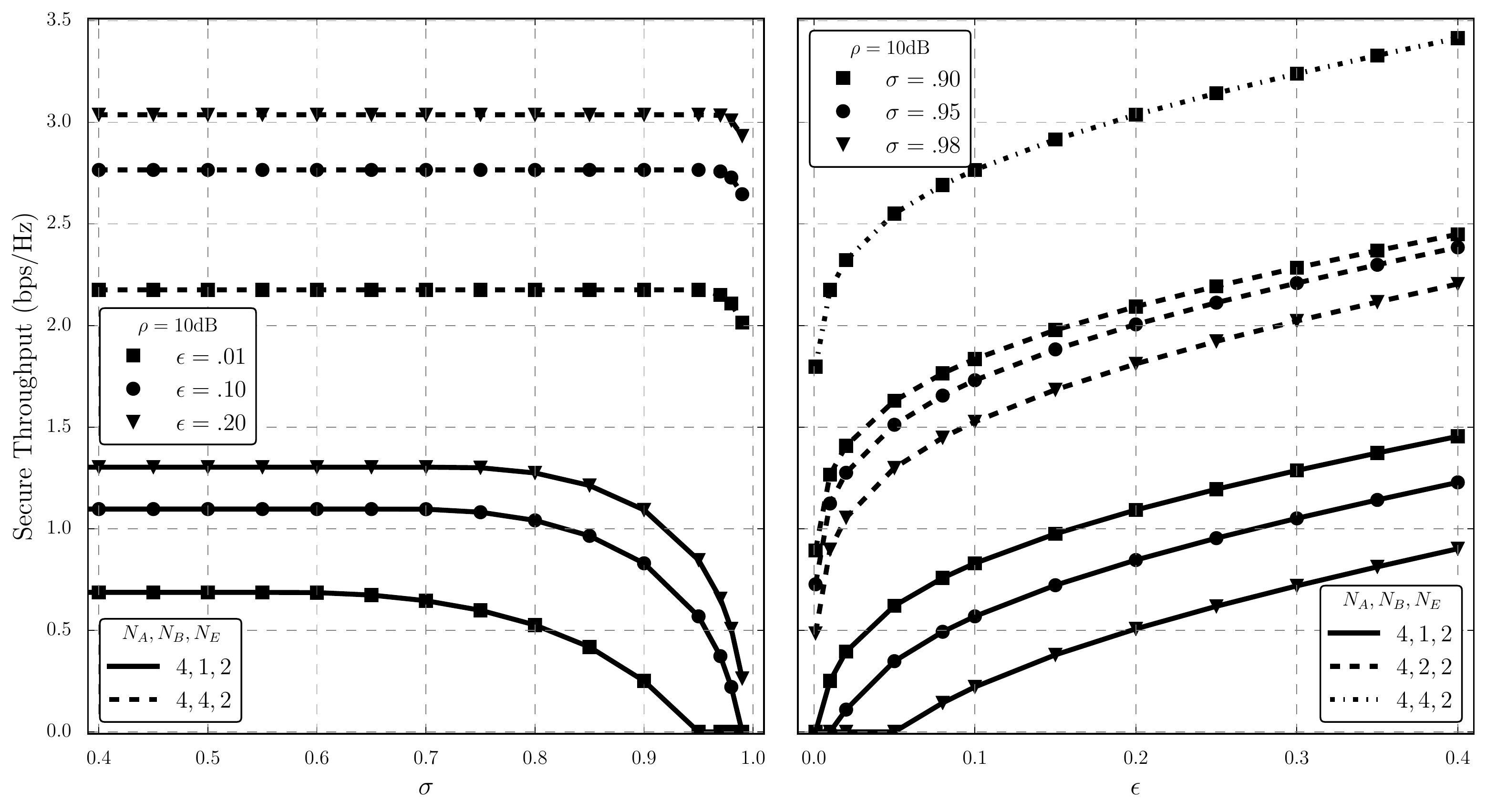}
		\caption{Secure throughput as a function of the $\sigma$ on the left and $\epsilon$ on the right, for fixed $\overline{\gamma}_B=10$ dB, assuming $\overline{\gamma}_E=0$~dB, and distinct antenna configurations.}
		\label{fig:tput_s_e_g}
	\end{figure*}
	
	We illustrate Proposition~\ref{prob:opt_rb} with the following numerical example depicted in Fig.~\ref{fig:tput_gamma}, where secure throughput is evaluated as a function of the SNR of the legitimate link $\overline{\gamma}_B$, assuming $\overline{\gamma}_E=0$~dB (thus $\rho = \overline{\gamma}_B$), $\sigma=95\%$ for $N_A=4$, $N_B\in\{1,4\}$ and $N_E=2$. As expected, by relaxing the constraint on the secrecy outage $\epsilon$, larger throughput can achieved. 
	Similar effect can be also observed if we relax the QoS constraint ($\sigma$). An significant improvement can be observed as the number of antennas at the legitimate channel grows. 

	Fig.~\ref{fig:tput_s_e_g} further shows the secure throughput as a function of the legitimate link QoS ($\sigma$) on the left, and as a function of the secrecy outage threshold ($\epsilon$) on the right. We assume $\rho=10$~dB and distinct antenna configurations. Again, the higher the number of antennas at the legitimate link, greater throughput is achieved, for instance by increasing by one the number of antennas at Bob the throughput more than doubles ($\times 2.33$) for $\epsilon=.10$ and $\sigma=.90$. 
	
	Such throughput enhancement can be also observed by relaxing the constraint on the secrecy outage. Interestingly, some performance floors are achieved with respect to our constraints.
	For instance, from Fig.~\ref{fig:tput_s_e_g} on the right, the QoS constraint can be relaxed from $98$\% to $90$\% and yet the same throughput is achieved even  having more antennas available at the legitimate channel, which is consequence of Proposition~\ref{prop:sec_reliability}.

	All in all, our results show the trade-off between security and reliability, and thus depending on the application more relaxed secrecy constraints can ensure great reliability. Likewise, we are also able to trade reliability for secrecy, which in this case goes against current standards for smart grids that requires at least $98$\% reliability in the communication link \cite{ART:Kuzlu-CN-2014}. Nonetheless, \cite{ART:Kuzlu-CN-2014} do not account for security and herein we show that such constraint can be achieved and even higher security can be guaranteed if the reliability constraint is relaxed. 

	In the following section, we illustrate our framework with a practical smart grid example.
	As we will see later, the information to be transmitted is average power demand of a household where the aggregator and the eavesdropper need to reconstruct the load profile curve. 
	
	\section{Secure Reconstruction of the Average Power Demand Curve} \label{sc:power}
	In the previous section, we commented that if we allow a larger secrecy outage, higher throughput in the legitimate link can be attained. Consequently, larger secrecy outage means that the Eve will decode more information and become more knowledgeable about our system. However, as we discuss next, even in this case Eve will not be able to acquire enough information to reconstruct the average power demand curve completely. 
	
	\begin{figure*}[!t]
		\centering
		\includegraphics[width=\columnwidth]{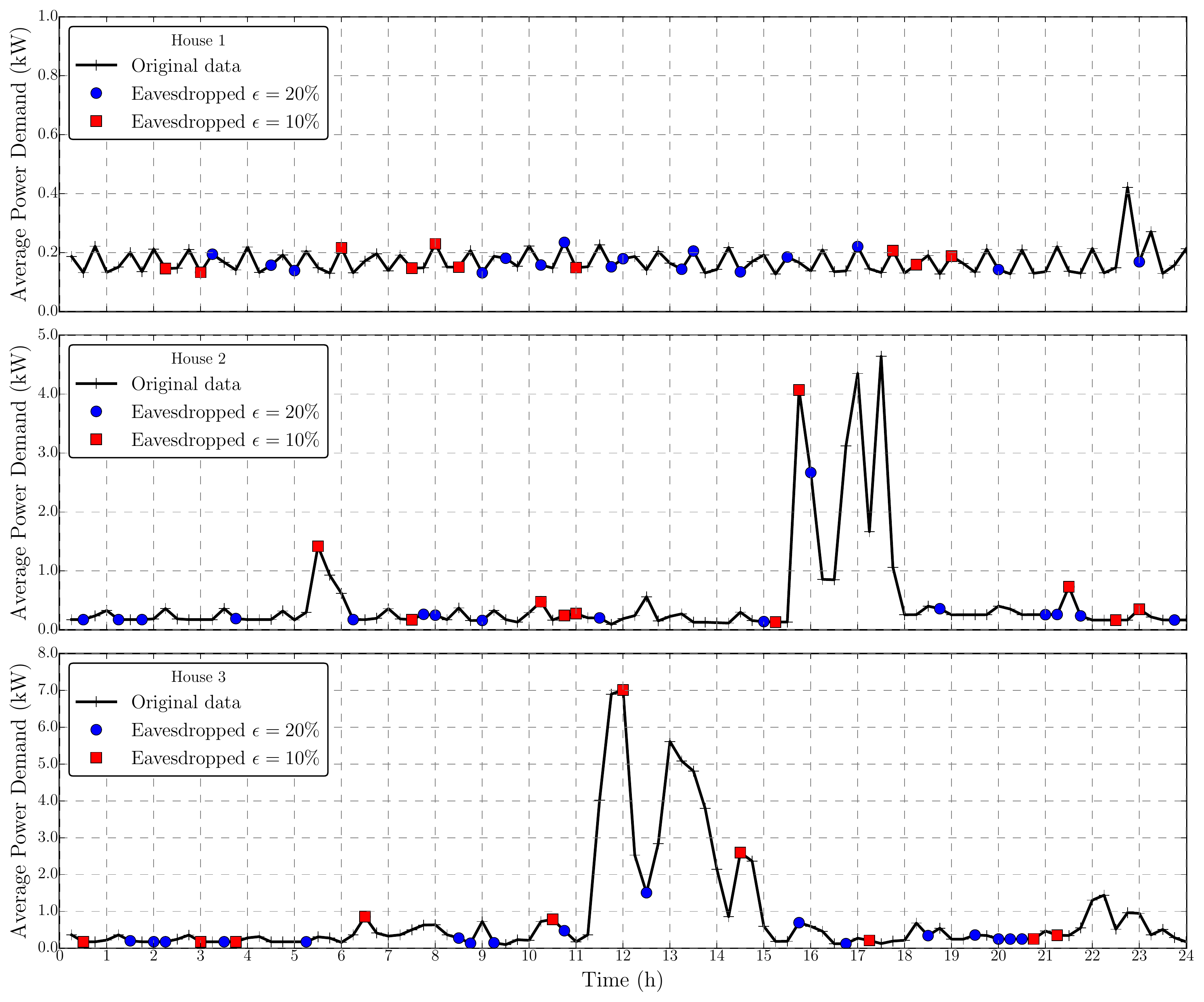}
		\caption{Examples of average power demand curves for three distinct houses from REDD database over 24 hours with transmissions every 15 minutes. House~1 presents a low power demand (few appliances (e.g. fridge) are on), House~2 has higher average and presents peak demand, which is also observed in House~3. The cases where the eavesdropper acquires 10\% and 20\% of the packages if also depicted.}
		\label{fig:houses}
		\vspace{1ex}
	\end{figure*}

	Let us first exemplify how Alice transmits its average power demand to Bob. Then, let $x[n]$ denote the average power demand, where $n=1,...,N$ is the index and $N$ total number of samples. 
	Transmissions are schedule in fixed period of time $\tau$, herein we assume 15-minute based sampling and a transmission and thus $\tau=0.25$ hour, which renders $N=96$ samples per day. 
	If an outage occurs, Bob reconstructs the signal via linear interpolation between two adjacent points. Thus, Bob will interpolate the missing value(s) using the latest two received samples. Similarly, we assume that Eve also attempts to estimate and reconstruct the signal via linear interpolation. 
	For example, consider the transmitted sequence: $x[k-2], x[k-1], x[k]$  with $k=2,...,N$, and let $y[k]$ denote the received signal, $k=1,...,N$. Then, if the samples $x[k-2]$ and $x[k]$ are successfully received but $x[k-1]$ is not, the reconstruction is based on the linear interpolation and the estimated point is denoted by $y[k-1] = (y[k] + y[k-2])/2$. 
	
	In order to perform our analysis we resort to ``The Reference Energy Disaggregation Data Set'' (REDD) database \cite{PROC:Kolter-2011,URL:REDD} to build the signal $x[n]$, which is a 15-minute average power demand over a timespan of $24$ hours\footnote{The REDD database is composed of 6 households, monitored during several days with a frequency of $1$Hz. After processing the data (namely, the sum of the power of phases A and B), we identified $53$ slices of $24$-hour periods (all aligned in time among themselves) which provide us a full set of average power measures. In other words, each of these slices can be seen as a single household and then these measures are used to simulate the daily transmissions}.
	We assume that both Bob and Eve use linear interpolation to reconstruct the power demand curve. In order to estimate the error due to the signal reconstruction we adopt the root mean square deviation (RMSD), which is calculated based on the received (and estimated when needed) samples and the actual data, and is given as
	\begin{eqnarray}\label{eq:rmsd}
	\mathrm{RMSD} = \sqrt{\frac{1}{N}\sum_{k=1}^{N}(y[k]-x[k])^2}.
	\end{eqnarray}
	In order to facilitate the comparison among household power demand profiles, we choose to normalize the RMSD (NRMSD) by the average of the transmitted signal power, thus $\mathrm{NRMSD} = {\mathrm{RMSD}}/{\overline{y}}$, which is commonly known as the coefficient of variation of the RMSD.
	
	Further, from the database selected $3$ households that provide a significant representation of the dataset, namely House~$1$, House~$2$ and House~$3$, since each of these households presents a distinct average power demand profile. Fig.~\ref{fig:houses} exemplifies the average power demand of these three distinct houses over 24 hours with transmissions every 15 minutes. 
	For instance, House~1 presents a low power demand profile, which means that few appliances are on (e.g. fridge, lights). 
	House~2 has higher average compared to House~1 and presents peak demand, which is also observed in House~3. For instance, from the data of House~2 we can infer that there is more activity in the house in early morning (e.g. showering, preparing breakfast) and at the end of the day, around the time where people are having dinner, doing the house chores, and watching TV. House~3 has similar patterns, but shifted in time and concentrated during the afternoon. 
	Fig.~\ref{fig:houses} also assumes that an eavesdropper is able to decode  10\% (red square) or 20\% (blue circle) of the packages, due to the secrecy outages occurred in this period. Note that if Eve can obtain 20\% of the packets, she is still not able to reconstruct the power demand curve and then infer the presence and activities within a given house. 
	On the other hand, with $\sigma \geq 90\%$ few points are lost such that the aggregator can estimate them through linear interpolation without larger estimations errors. 
	It is noteworthy that a malicious eavesdropper may acquire this information and perform a series of cyber (and even physical) attacks on a neighborhood by exploiting the smart meters transmissions. Given enough intercepted points, it is  possible to infer personal information and inhabitants behavior and activities (for instance, presence and absence hours, sleeping hours) from the power demand curve \cite{ART:Fang-CST2012}. 
	Thus the necessity of protecting the transmission against eavesdropping and any other type information leakage.  
	
		\begin{figure}[!t]
			\centering
			\includegraphics[width=\columnwidth]{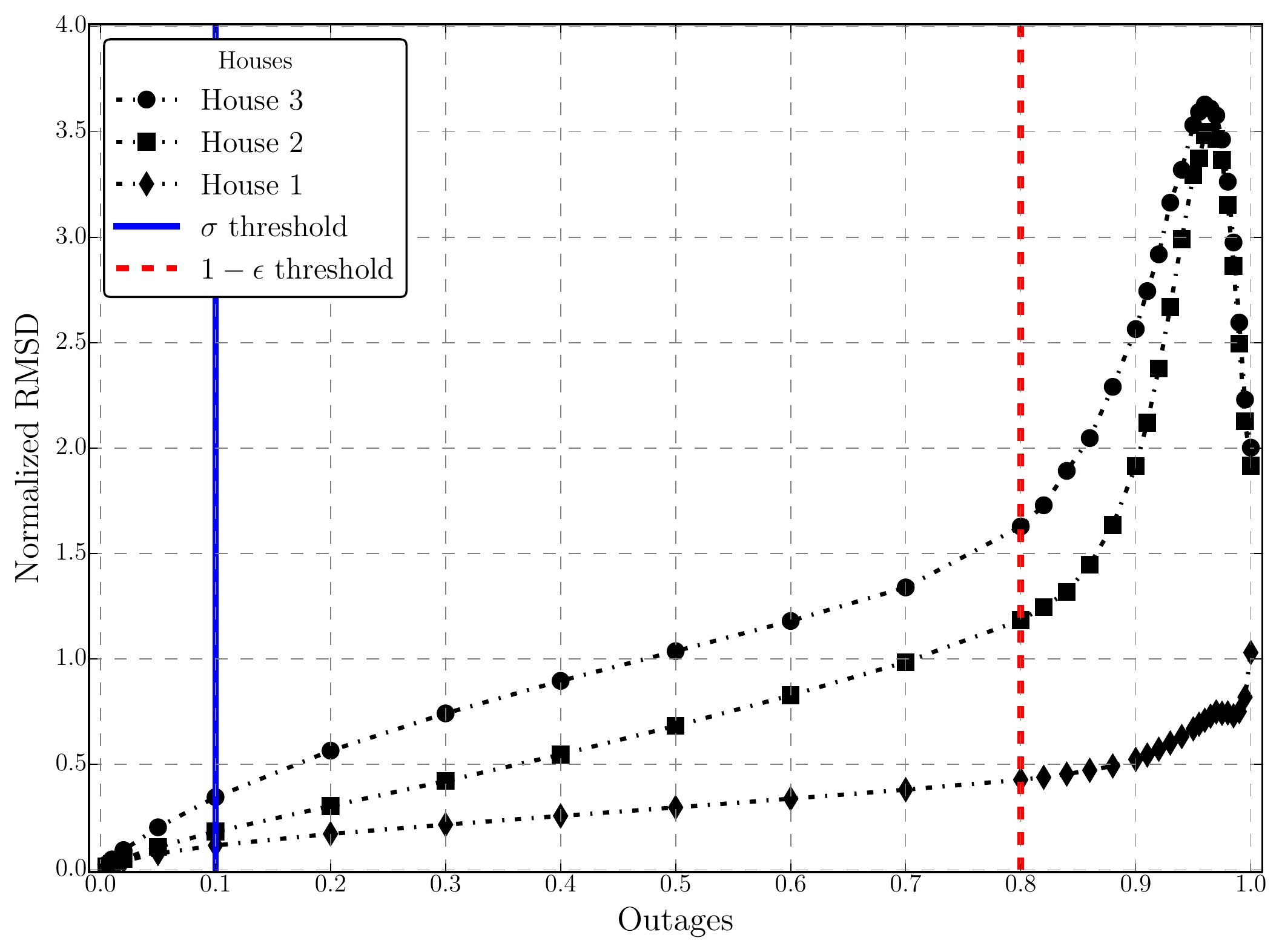}
			\caption{Normalized RMSD as a function of the outage, which encompass reliability and secrecy outages. Note that $\sigma$ (blue line on the left) and $1-\epsilon$ (dashed red line on the right) thresholds delimit the regions that guarantee secrecy and reliability.}
			\label{fig:error}
		\end{figure}
	Fig.~\ref{fig:error} depicts the normalized RMSD as a function of the outages, which represents the outage either at the legitimate link or information leakage to Eve. In terms of reliability, the region of interest lies on the left-hand side of the plot and it is delimited by the $\sigma$ threshold (blue line). We assume Monte Carlo simulations with $10^5$ repetitions for each house (each house has $N=96$ samples). After the linear interpolation used to estimate the missing points, we calculate the RMSD and then normalize by the average power. 
	As we can observe from the figure, if during a day the legitimate link perceives outages of up to $\sigma=10$\%, the demand power curve can be reconstructed with low error. For instance, for $\sigma=10$\%, the normalized RMSD for each house is respectively $0.12$, $0.18$ and $0.35$, and these values can be seen as coefficient of variation indicating that there is a low variance in the reconstruction of the power demand curve.
	Eve, on the other hand, has a greater outage, which increases the error in the signal reconstruction leading to high coefficient of variation. For example, secrecy outages of at most $\epsilon \leq 20$\% (which means that, on average,  Eve intercepts up to $20$\% of the transmissions) correspond to outages greater than $80$\% and therefore a higher coefficient of variation of the RMSD as indicated by red line ($1-\epsilon$ threshold) on the rightmost side of Fig.~\ref{fig:error}. Notice that Eve acquires few points, and thus her estimation and reconstruction is very poor, as a result of Alices strategy when setting the secrecy outage threshold and optimizing the secure throughput. 
	It is worth mentioning that the reduction on the normalized RMSD on the extreme right (more than $97$\% of outages) occurs because the number of points available at Eve is small. In this cases the RMSD is calculated with respect to zero or to a line that lies close to average of the actual signal, which decreases the RMSD. To illustrate this point assume that Eve only attained a point around hours $3$ and $21$ from House $3$, as depicted in Fig.~\ref{fig:houses}. Based only on that, Eve estimates that all points lie within this line, and as we can see from Fig.~\ref{fig:houses} the majority of the points is closer to the marginal power demand rather than to the peak consumption hours.

	\section{Discussions} \label{sc:discussions}
	
	We proposed a physical layer security scheme that enhances the communication link between a pair of legitimate nodes in the presence of an eavesdropper. In our scenario, an eavesdropper may attempt to acquire information from the smart meters from a given neighborhood. However, the transmitter does not have any CSI from the Eve, but is still able to optimize its transmission rate such that secure throughput can be achieved. 
	Notice that the results attained herein are not limited to smart grid applications, thus we provide an general framework that can be extended to other contexts. 
	In the previous sections we have discussed how we can improve the secure throughput, reliability and security of the system and we connect our analytical results with actual data and signal reconstruction. Herein, we discuss some pros and cons of this proposed method and future work. 
	
	\subsection{Reliability and secrecy outage constraints}
	We set the reliability constraint ($\sigma$) to ensure a minimum robustness for the legitimate link. Likewise, the secrecy outage constraint ($\epsilon$) envisages a maximum information leakage to the eavesdropper. 
	Then, we present the trade-off between security and reliability, in which we can choose to sacrifice robustness of the legitimate link for security, or relax the secrecy constraint in order to achieve higher reliability. 
	Current standards foresee a reliability greater than $98\%$ for the communication link in the smart grid (smart meter--aggregator) \cite{ART:Kuzlu-CN-2014}. As discussed above, such constraint is stringent especially if we want add security at physical layer while enhancing the performance of the system. 
	
	Then, \textit{how big would be the sacrifice of robustness of the legitimate link for security?} 
	A more appropriate answer can be given only if we know the information that is sent to the
	aggregator, so that different information flows have distinct priorities and allocation. In our case, we assume that the information sent is the average power demand, then we show that the signal reconstruction is possible even with relatively loose outage constraints (e.g. $\sigma\geq 90\%$), while Eve cannot attain much information at secrecy outages of $\epsilon \leq 10\%$. 
	Notice that this result is dependent on the inherent characteristic of the transmitted signal, for instance, as we can observe from Fig.~\ref{fig:houses} the average power demand presents overall low variation around the average (see House~$1$), except for relatively short periods of peak consumption as in House~$2$ and House~$3$.
	Therefore, design the whole system for higher outage probabilities in the legitimate link as well as high secrecy outages may not be prudent for other kind of signals or if the aggregator should provide
	some kind of feedback to the smart meter (e.g. change the power demand behavior, as in strategies
	of demand-side management \cite{ART:Strbac-EP2008}). Thus, the necessity of classifying the information flows from smart meters to aggregator with respect to signal characteristics, as well as reliability and security. 
	
	\subsection{Enhancing robustness of the legitimate link }
	
	Herein, we assume that transmissions are scheduled every $15$ minutes, and if a package does not meet our outage constraint, it is considered lost and then Bob will estimate the power demand via interpolation. Alternatively, as future work, another scenario may include Hybrid Automatic Repeat Request (HARQ) strategies in order to enhance the communication link, which reduces outage events while enhancing throughput as show in \cite{ART:Nardelli-TWC2012, ART:Alves-WN2012}. Cooperation may also be an extension to enhance secure throughput and reliability \cite{ART:Alves-SPL-2015, ART:Brante-TC-2015}. 
	Thus, these more advanced communication techniques combined may be used to enhance secure throughput and would be an interesting next step for the present work.
	
	\subsection{Signal processing and transmission}
	Fig. \ref{fig:houses} exemplifies a $15$ minute sampling interval of the power demand of a household.
	Due to the characteristics of this signal, a time-based sampling might not be the most effective way to collect and sent data to the aggregator. Then, as pointed out in \cite{ART:Simonov-SJ2015} event based sampling may be more suitable, once it reduces the amount of redundant information transmitted. 
	However, such approach requires a more robust communication link, due to the lack of redundant data,  and therefore the loss of any sample will have a more dramatic effect on the signal reconstruction. At the same time, this scheme is more secure once Eve acquires even less information. 
	As pointed out in \cite{ART:Nardelli-AN2015} transmission strategies and outage constraints should be evaluated in combination with the sampling procedure, due to the amount of redundant information generated in each case. 
	
	It is worth noting that even though we analyze the situation for a 24-hour period and a simple interpolation technique. Due to the daily habits of the dwellers, it would be possible to recover the usage profile (or activities) by superimposing the missing data from one day with data from similar days. However, this would require more sophisticated signal processing at Eve as well as large time window that could range from days to months depending on the settings of the network. Such process is also hampered by slight variations in the habits and activities of the inhabitants if we consider a sufficiently high outage for Eve.

	\section{Conclusions and final remarks} \label{sc:conc}
	
	Herein we assess the secure communication link between smart-meters and an aggregator in the presence of a potential eavesdropper (Eve). We assume MIMOME wiretap channel, where Alice employs transmit antenna selection and has no channel state information of Bob and Eve. Therefore, we resort to fixed Wyner codes and then optimize Alice's transmission rate so that secure throughput can be guaranteed subject to quality of service and secrecy outage constraints. 
	We assess in closed-form both secrecy outage and successful transmission probabilities, and then maximize the secure throughput and establish the secrecy-reliability trade-off. Our numerical results illustrate the effect of this trade-off on the secure throughput as well as number of antennas at Alice, Bob and Eve. Our results show that a small sacrifice in reliability allows secrecy enhancement. Even though Eve may acquire some information, we show that it will not be enough to reconstruct the average power demand curve of a household.
	
	We plan to study in future works how the secure throughput will be affected under different sampling strategies. In this way, we plan to build a joint sampling-transmission technique that can improve the system efficiency, as discussed in Section~\ref{sc:discussions}.

	
\end{document}